\documentclass[11pt,fleqn]{article}

\usepackage{amsmath,amssymb,amsthm}
\usepackage{fullpage}
\usepackage[hypertex]{hyperref}

\DeclareMathOperator{\ddim}{ddim}
\DeclareMathOperator{\diam}{diam}
\newcommand{\veps}{\varepsilon}

\newcommand{\etal}{{\it{et al.\ }}}

\newtheorem{theorem}{Theorem}[section]
\newtheorem{lemma}[theorem]{Lemma}

\newtheorem{corollary}[theorem]{Corollary}

\newtheorem{fact}[theorem]{Fact}

\title{Approximate nearest neighbor search for $\ell_p$-spaces $(2<p<\infty)$ via embeddings}
\author{
Yair Bartal
\footnote{Hebrew University. Work supported in part by an Israel Science Foundation grant \#1609/11.
Email: \texttt{yair@cs.huji.ac.il} }
\and
Lee-Ad Gottlieb
\footnote{
Ariel University. Work supported in part by an Israel Science Foundation grant \#755/15.
Email: \texttt{leead@ariel.ac.il} }
}

\begin{document}

\maketitle

\begin{abstract}
While the problem of approximate nearest neighbor search has been well-studied 
for Euclidean space and $\ell_1$, few non-trivial algorithms are known for 
$\ell_p$ when $2<p<\infty$. In this paper, we revisit this fundamental problem and 
present approximate nearest-neighbor search algorithms which give the first 
non-trivial approximation factor guarantees in this setting.
\end{abstract}

\section{Introduction}

Nearest neighbor search (NNS) is one of the basic operations computed 
on data sets comprising numeric vectors, i.e.\ points. The problem asks to preprocess 
a $d$-dimensional set $V$ of $n=|V|$ vectors residing in a certain space $M$, so that 
given a new query point 
$q \in  M$, a point nearest to $q$ in $V$ can be located efficiently. This problem has 
applications in data mining, database queries and related fields.

When the ambient space $M$ is a high-dimensional $\ell_p$-space,\footnote{This is
a space equipped with a Minkowski norm, which defines the distance between two
$d$-dimensional vectors $x,y$ as $\|x-y\|_p = (\sum_{i=1}^d |x_i-y_i|^p)^{1/p}$.}
NNS may require significant 
computation time, and this is due to the inherent complexity of the metric. 
For example, for Euclidean vectors ($\ell_2$-space), Clarkson \cite{C88} gave
an $O(n^{\lceil d/2 \rceil (1 + \veps)})$ size data structure that answers exact
NNS queries in $O(\log n)$ time (with constant factors in the bounds depending on 
constant $\veps>0$), and claimed that the exponential dependence on $d$ is a manifestation
of Bellman's \cite{Be61} ``curse of dimensionality.'' 
This has led researchers to consider the $c$-approximate
nearest neighbor problem (ANN), where the goal is to find a point in $V$ 
whose distance to $q$ is within a factor $c$ of the 
distance to $q$'s true nearest neighbor in $V$.
In the Euclidean setting, celebrated results 
of Kushilevitz, Ostrovsky and Rabani \cite{KOR98} (see also \cite{OR00,OR02}) 
and Indyk and Motwani \cite{IM98,HIM12}
achieved polynomial-size data structures which return a $(1+\veps)$-ANN
in query time polynomial in $d \log n$ (when $\veps > 0$ is any constant).
These results can be extended to all $\ell_p$ with $1 \le p \le 2$.

However, the more difficult regime of $p>2$ is significantly less well understood. 
Recalling that for any vector $v$ and $p>2$, 
$d^{\frac{1}{p} - \frac{1}{2}} \|v\|_2 \le \|v\|_p \le \|v\|_2$, 
we conclude that simply running an $\ell_2$ ANN algorithm on $V \subset \ell_p$
(that is, treating $V$ as if it resided in $\ell_2$) will return an
$O(d^{\frac{1}{2} - \frac{1}{p}})$-ANN in polylog time.
If we allow exponential space, then a $(1+\veps)$-ANN can be found in query time $O(d \log n)$ 
and space $\veps^{-O(d)}n$ by utilizing an approximate Voronoi diagram \cite{HIM12,AM02}.
(These structures were developed for Euclidean spaces, but apply to all $\ell_p$, $p \ge 1$ as well \cite{HK15}.)
For $\ell_\infty$, Indyk \cite{In98} gave a polynomial-size structure which answers
$O(\log \log d)$-approximate queries in $d \log^{O(1)} n$ query time, 
and remarkably there are indications that this bound may be optimal \cite{ACP08}. 
Since for any vector $v$ we have $\|v\|_\infty \le \|v\|_p \le d^{1/p} \|v\|_\infty$,
Indyk's $\ell_\infty$ structure gives a $O(d^{1/p} \log\log d)$-ANN algorithm for all $p>2$,
and in particular an $O(\log\log d)$-ANN whenever $p = \Omega(\log d)$.

\paragraph{Our contribution.}
We revisit the problem of ANN in $\ell_p$ spaces for $2 < p < \infty$, and give improvements
over what was previously known. 
Note that the $\ell_p$-norm for this regime finds application in fields such as
image resolution \cite{FT04,FRT11}, time-series comparison \cite{YYW11}, and $k$-means clustering \cite{CM12}.
We are interested in polynomial-size structures that have query
time polynomial in $d \log n$. Hence we shall make the simplifying assumptions that 
$d = \omega (\log n)$ and $d = n^{o(1)}$: If 
$d = O(\log n)$ then approximate Voronoi diagrams may be used, and if
$d = n^{\Omega(1)}$ then comparing the query point $q$ to each point in $V$ 
in a brute-force manner can be done in time
$O(dn) = d^{O(1)}$.\footnote{We recall also that there exists an oblivious mapping for 
all $\ell_p$ that embeds $\ell_p^m$ into $\ell_p^d$
for $d = {n \choose 2}$ dimensions \cite{Fi88,Ba90}.}

Our first result is an extension of Indyk's $\ell_{\infty}$ algorithm to
smaller values of $p$. Exploiting the max-stability property of Fr\'{e}chet random variables,
we give in Section \ref{sec:p-logn} a randomized embedding from $\ell_p$ into $\ell_\infty$
that is non-contractive and has expansion $O(\log^{1/p}n)$.
This means that points distant from query point $q$ remain far, while the distance from $q$
to its nearest neighbor increases by at most a factor of $O(\log^{1/p}n)$. 
We then run Indyk's algorithm on the embedded space, and the result must be a
$O ( \log\log d \log^{1/p} n )$-ANN in the origin space.
We refine this technique in Section \ref{sec:p-ddim} to obtain better bounds
for spaces with low doubling dimension.

Having extended the range in which Indyk's $\ell_\infty$ algorithm is applicable, we
proceed to introduce an embedding which greatly extends the range for which the $\ell_2$ 
algorithms are applicable. In Section \ref{sec:mazur}, we 
introduce the Mazur map as an algorithmic tool. This mapping allows us to embed $\ell_p$ into
$\ell_2$, and we then solve ANN in the embedded space. Although the Mazur map induces 
distortion dependent on the diameter of the set, thereby confounding the ANN search, 
we show that the mapping can be applied to small low-diameter subproblems.
Our final result is a polynomial-size structure which answers $2^{O(p)}$-approximate queries in time 
polynomial in 
$d \log n$ (Theorem \ref{thm:nns-big}). This yields non-trivial results for $p = o(\log d)$. 
Comparing this result with the one above: 
\begin{itemize}
\item
When $p = O(\sqrt{\log\log n})$,
the $2^{O(p)}$-ANN algorithm is best. 
\item
When
$p = \Omega(\sqrt{\log\log n})$ the 
$O ( \log\log d \log^{1/p} n )$-ANN algorithm
is best.
\end{itemize}
Note that the worst case is when 
$p = \Theta(\sqrt{\log\log n})$, where 
the approximation ratio is 
$2^{O(p)} = 2^{O(\sqrt{\log\log n})}$.
As we have assumed 
$d = \omega (\log n)$, the above bound is $2^{o(\sqrt{\log d})}$, much better
than the previously known $d^{O(1)}$-factor approximations.

\subsection{Related work}
For Euclidean space, Chan \cite{Ch98} gave a deterministic construction which gives 
an $O(d^{3/2})$-ANN, in time $O(d^2 \log n)$ and using polynomial space (see also \cite{Be93}). 
For $\ell_p$, Neylon \cite{Ne10} gave an
$O(d)$-ANN structure which runs in $O(d^2 \log n)$ time and uses $\tilde{O}(dn)$ space.

Stoev \etal \cite{SHKT07} and Andoni \cite{An13} both used random variables 
with max-stability or min-stability to estimate the $p$-th moment of a vector, or of the 
difference between two vectors. 
Farag\'{o} \cite{Fa11} presented an elegant oblivious embedding 
from $\ell_p^d$ to $\ell_\infty^{2^{O(d)}}$
with arbitrarily low distortion, and the existence of an embedding with these 
properties had been alluded to in Indyk's survey \cite{In01}.
We observe that one may utilize Farago's embedding to map $\ell_p$ into $\ell_\infty$ and then 
compute the nearest neighbor in the embedded space using Indyk's $\ell_{\infty}$ structure. 
This in fact yields an $O(\log d)$-ANN, 
but in time and space exponential in $d$, so approximate Voronoi diagrams are better for
this problem.

For ANN in general metric spaces, Krauthgamer and Lee \cite{KL04} showed that the doubling dimension
can be used to control the search runtime: For a metric point set $S$, they constructed 
a polynomial-size structure which finds an $O(1)$-ANN in time $2^{O(\ddim(S))} \log \Delta$,
where $\Delta = \Delta(S)$ is the {\em aspect ratio} of $S$, the ratio between the 
maximum and minimum inter-point distances in $S$. The space requirements of this data
structure were later improved by Beygelzimer \etal \cite{BKL06}. Har-Peled and Mendel \cite{HM06} 
and Cole and Gottlieb \cite{CG06} showed how to replace the dependence on $\log \Delta$ with
dependence on $\log n$.
Other related research on nearest neighbor searches have focused on various
assumptions concerning the metric space. Clarkson \cite{Cl99} made
assumptions concerning the probability distribution from which the database and query points 
are drawn, and developed two randomized data structures for exact nearest
neighbor. Karger and Ruhl \cite{KR02}
introduced the notion of growth-constrained metrics (elsewhere called the
KR-dimension) which is a weaker notion than that of the doubling
dimension. A survey of proximity searches in metric space appeared in
\cite{CNBYM01}.

Subsequent to the dissemination of the results in this paper, we were advised of a manuscript of
Naor and Rabani \cite{NR06} (mentioned in \cite[Remark 4.2]{N14}) which gives a similar 
$2^{O(p)}$-ANN algorithm for $p>2$, also utilizing the Mazur map.
In personal communication, Assaf Naor broached the question of better dependence on 
$p$ in the $2^{O(p)}$ approximation bound of Theorem \ref{thm:nns-big}. He noted that all
{\em uniform} embeddings of $\ell_p$ ($p>2$) into $\ell_2$ 
(such as the Mazur map) possess distortion exponential in $p$
\cite[Lemma 5.2]{N14}.
Non-uniform embeddings of $\ell_p$ into $\ell_2$ may possess better distortion bounds.

\subsection{Preliminaries}\label{sec:prelim}

\paragraph{Embeddings and metric transforms.}
A much celebrated result
for dimension reduction is the well-known $l_2$ flattening lemma of Johnson
and Lindenstrauss \cite{JL84}: For every $n$-point subset of $l_2$ and every
$0<\veps<1$, there exists a mapping into $l_2^k$ that preserves all inter-point
distances in the set within factor $1+\veps$, with target dimension
$k = O(\veps^{-2} \log n)$.

Following Batu \etal \cite{BES06}, we define an {\em oblivious} embedding to be an embedding
which can be computed for any point of a database or query set,
without knowledge of any other point in these sets. (This differs slightly from
the definition put forth by Indyk and Naor \cite{IN07}.) Familiar oblivious embeddings 
include standard implementations of the Johnson-Lindenstrauss transform 
for $l_2$ \cite{JL84}, the dimension reduction mapping of Ostrovsky and Rabani \cite{OR02} 
for the Hamming cube, and the embedding of Johnson and Schechtman 
\cite{JS82} for $\ell_p$, $p \le 2$.

An embedding of $X$ into $Y$ with {\em distortion} $D$ is a mapping $f:X\rightarrow Y$ 
such that for all $x,y \in X$,
$$ 1 \leq c \cdot \frac{d_Y(f(x),f(y))}{d_X(x,y)}   \leq D,$$
where $c$ is any scaling constant.
An embedding is {\em non-contractive} if for all $x,y \in X$,
$$ \frac{d_Y(f(x),f(y))}{d_X(x,y)}   \geq 1.$$
It is {\em non-expansive} (or {\em Lipschitz}) if for all $x,y \in X$,
$$ \frac{d_Y(f(x),f(y))}{d_X(x,y)}   \leq 1.$$

\paragraph{Doubling dimension.}
For a metric $M$, let $\lambda>0$
be the smallest value such that every
ball in $M$ can be covered by $\lambda$ balls of half the radius.
$\lambda$ is the {\em doubling constant} of $M$, and
the {\em doubling dimension} of $M$ is $\ddim(M)=\log_2\lambda$.
Then clearly $\ddim(M) = O(\log n)$.
Note that while a low $\ell_p$ vector dimension implies a low doubling 
dimension -- simple volume arguments demonstrate that $\ell_p$ metrics 
($p \ge 1$) of dimension $d$ have doubling dimension $O(d)$ -- 
low doubling dimension is strictly more general than low $\ell_p$ dimension.
We will often use the notation $\ddim$ when the ambient space is clear from context.
The following packing property can be shown (see for example \cite{KL04}):

\begin{lemma}
\label{lem:pack}
Suppose that $S\subset M$ has a minimum inter-point distance
$\alpha$, and let $\diam(S)$ be the diameter of $S$. Then
$$ |S| \leq \Big(\tfrac{2\diam(S)}{\alpha}\Big)^{\ddim(M)}. $$
\end{lemma}

\paragraph{Nets and hierarchies.}
Given a point set $S$ residing in metric space $M$, 
$S' \subset S$ is a $\gamma$-net of $S$ if the minimum inter-point distance in 
$S'$ is at least $\gamma$, while the distance from every point of $S$ to its nearest neighbor in 
$S'$ is less than $\gamma$. Let $S$ have minimum inter-point distance 1.
A {\em hierarchy} is a series of $\lceil \log \Delta \rceil$ nets ($\Delta$ being the aspect ratio 
of $S$), where each net $S_i$ is a $2^i$-net of the previous net $S_{i-1}$.
The first (or bottom) net is $S_0 = S$,
and the last (or top) net $S_t$ contains a single point called the {\em root}.
For two points $u \in S_i$ and $v \in S_{i-1}$, if $d(u,v)<2^i$ then we say that $u$ {\em covers} $v$, and
this definition allows $v$ to have multiple covering points in $2^i$.
The closest covering point of $v$ is its {\em parent}.
The distance from a point in $S_i$ to its ancestor in $S_j$ $(j>i)$ is at most
$\sum_{k=i+1}^j 2^k = 2 \cdot (2^j - 2^{i+1}) < 2 \cdot 2^j$.

Given $S$, a hierarchy for $S$ can be built in time 
$\min \{ 2^{O(\ddim)}n, O(n^2) \}$, and this term also bounds the 
space needed to store the hierarchy \cite{KL04,HM06,CG06}. 
(The stored hierarchy is {\em compressed}, in that points which do not cover any other
points in the previous net may be represented implicitly.)
Similarly, we can maintain links from each
hierarchical point in $S_i$ to all neighbors in net $S_i$ within distance $c \cdot 2^i$, and this
increases the space requirement to $\min \{ c^{O(\ddim)}n, O(n^2) \}$.
From a hierarchy, a {\em net-tree} may be extracted by placing an edge between each point
$p \in S_i$ and its parent in $S_{i+1}$ \cite{KL04}. 
The height of the (compressed) tree is bounded by $O (\min \{n, \log \Delta \})$.

\section{Background: The near neighbor problem}\label{sec:background}

In this section, we review basic techniques for ANN in $\ell_p$ spaces.
A standard technique for ANN on set $V \subset \ell_p$ 
is the reduction of this problem to that of 
solving a series of so-called approximate {\em near} neighbor problems 
\cite{IM98,HIM12,KOR98} (also called the point location in equal balls problem).
The $c$-approximate near neighbor problem for a fixed distance $r$ is defined thus: 
\begin{itemize}
\item
If there is a point in $V$ within distance $r$ of query $q$, return some
point in $V$ within distance $cr$ of $q$.
\item
If there is no point in $V$ within distance $r$ of query $q$, return
{\em null} or some point in $V$ within distance $cr$ of $q$.
\end{itemize}
For example, suppose we had access to an oracle for the 
$c$-approximate near neighbor problem.
If preprocess a series of oracles for the $O(\log \Delta)$ values 
$$r = \{ \diam(V), \frac{\diam(V)}{2}, \frac{\diam(V)}{4}, \ldots \},$$
and query them all,
then clearly one of these queries would return a $2c$-ANN of $q$.
In particular, if $r'$ is the distance from $q$ to its nearest
neighbor in $V$, then the oracle query for the value for $r$ 
satisfying $r' \le r < 2r'$ would return such a solution.
Further, it suffices to seek the minimum $r$ that returns
an answer other then {\em null}. Then we may execute a binary
search over the candidate values of $r$, and so $O(\log \log \Delta)$
oracle queries suffice.

Har-Peled \etal \cite{HIM12} show that for all metric spaces, the 
ANN problem can be solved by making only $O(\log n)$ queries to
oracles for the near neighbor problem. The space requirement is
$O(\log^2n)$ times that required to store a single oracle.
The reduction incurs a loss in the approximation factor, but this loss
can be made arbitrarily small.
In Section \ref{sec:mazur}, we will require a more specialized reduction,
where we allow near neighbor oracle queries only on problem instances
that have constant aspect ratio.

\section{ANN for $\ell_p$-space via embedding into $\ell_\infty$}\label{sec:lp-linfty}

In this section, we show how to embed
$\ell_p$-space into $\ell_\infty$ in a way that 
guarantees that an ANN in the embedded space is
also an ANN in the origin space, albeit with a degradation in
the approximation quality.
This allows us to extend Indyk's $\ell_\infty$ ANN structure to $\ell_p$ as well. 
After giving the general result in Section \ref{sec:p-logn},
we refine it in Section \ref{sec:p-ddim} to give distortion that depends
on the doubling dimension of the space instead of its cardinality.

\subsection{Embedding into $\ell_\infty$}\label{sec:p-logn}

Here we show that any $n$-point $\ell_p^d$ space 
admits an oblivious embedding into
$\ell_\infty^d$ with favorable properties: 
The embedding is non-contractive with high probability, 
while the interpoint expansion is small. Hence the embedding 
approximately preserves the nearest neighbor for a fixed
query point $q$, and keeps more distant points far away.
This implies that Indyk's $\ell_\infty$ 
ANN algorithm can be applied to all $\ell_p$.

\paragraph{Max-stability.}
Our embedding follows the lead of \cite{SHKT07,An13} in utilizing max-stable 
random variables, specifically those drawn from a Fr\'{e}chet distribution: 
Having fixed $p$, our Fr\'{e}chet random variable $Z$ obeys for all $x > 0$
$$\Pr[Z \le x] = e^{-x^{-p}}$$
and so 
$$\Pr[Z > x] = 1 - e^{-x^{-p}} \le x^{-p}.$$
We state the well-known max-stability property of the Fr\'{e}chet distribution:

\begin{fact}\label{fact:max-stable}
Let random variables $X,Z_1,\ldots,Z_d$ be drawn from the above Fr\'{e}chet distribution,
and let $v = (v_1,\ldots,v_d)$ be a non-negative valued vector. Then the random variable
$$Y := \max_i \{ v_i Z_i \}$$
is distributed as $\|v\|_p \cdot X$
(that is, $Y \sim \|v\|_p \cdot X$).
\end{fact}

To see this, observe that 
\begin{eqnarray*}
\Pr[Y \le x] 
&=& \Pr[\max_i \{ v_i Z_i \} \le x]	\\
&=& \Pi_i \Pr[v_i Z_i \le x]		\\
&=& \Pi_i \Pr[Z_i \le x/v_i]		\\
&=& \Pi_i e^{-(v_i/x)^p}		\\
&=& e^{- (\sum_i v_i^p)/x^p}		\\
&=& e^{- (\|v\|_p/x)^{p}}.
\end{eqnarray*}
And similarly,
$\Pr[\|v\|_p \cdot X \le x]
= \Pr[X \le x/\|v\|_p]
= e^{- (\|v\|_p/x)^{p}}.$
So indeed the two random variables have the same distribution.

\paragraph{Embedding into $\ell_\infty$.}
Given set $V \subset \ell_p$ of $d$-dimensional vectors, 
the embedding $f_b: V \rightarrow \ell_\infty$ (for any constant
$b > 0$) is defined as follows:
First, we draw $d$ Fr\'{e}chet random variables $Z_1,\ldots,Z_d$
from the above distribution.
Embedding $f_b$ maps each vector $v \in V$ to vector
$f_b(v) = (b v_1 Z_1, \ldots, b v_d Z_d)$.
The resulting set is $V' \in \ell_\infty$.
Clearly, the embedding can be computed in time 
$O(d)$ per point.
We prove the following lemma.

\begin{lemma}\label{lem:p-logn}
For all $p \ge 1$, embedding $f_b$ applied to set $V \subset \ell_p$,
for $b = (3 \ln n)^{1/p}$, satisfies 
\begin{itemize}
\item
Contraction: $f_b$ is non-contractive
with probability at least $1-\frac{1}{n}$.
\item
Expansion: For any pair $u,w \in V$, 
$$\|f_b(u)-f_b(w)\|_\infty \le 2b \|u-w\|_p$$
with probability at least $1-2^{-p}$.
\end{itemize}
\end{lemma}

\begin{proof}
Consider some vector $v$ with $\|v\|_p = 1$.
Then by Fact \ref{fact:max-stable},
$\|f_b(v)\|_\infty \sim b \|v\|_p \cdot X = b \cdot X,$
where $X$ is a Fr\'{e}chet random variable
drawn from the above distribution. Then 
$$\Pr[\|f_b(v)\|_\infty < 1] 
= \Pr[b \cdot X < 1]
= e^{-(1/b)^{-p}}
= \frac{1}{n^3}.$$
Since the embedding is linear, $v$ may be taken to be an inter-point
distance between two vectors in $V$ ($v = \frac{u-w}{\| u-w \|_p}$), and so the  
probability that {\em any} inter-point distance decreases is less than
$n^2 \cdot \frac{1}{n^3} = \frac{1}{n}$. 

Also,
$\Pr[\|f_b(v)\|_\infty > 2 b] \le 2^{-p}$,
and so for any vector pair $u,w \in V$ we have
$\Pr[\|f_b(u)-f_b(w)\|_\infty > 2 b \|u-w\|_p] 
\le 2^{-p}.$
\end{proof}

Indyk's near neighbor structure is given a set $V \in \ell_\infty$ and distance $r$,
and answers $O(\log \log d)$-near neighbor queries for distance $r$ in time 
$O(d \log n)$ and space $n^{1+\delta}$, where $\delta$ is an arbitrarily small constant
(that affects the exact approximation bounds).
Combining this structure and Lemma \ref{lem:p-logn}, we have:

\begin{corollary}\label{cor:p-logn}
Given set $V \subset \ell_p$ for $p>2$ and a fixed distance $r$,
there exists a data structure of size $n^{1+\delta}$ 
(for arbitarily small constant $\delta$) 
that solves the 
$O(\log\log d \log^{1/p}n)$-approximate
near neighbor problem for distance $r$ with query time
$O(d \log n)$, 
and is correct with probability at least $1 - \frac{1}{n} - \frac{1}{2^p}$.
\end{corollary}

\begin{proof}
Given a set $V$ and distance $r$, we preprocess the set by computing the 
embedding of Lemma \ref{lem:p-logn} for each point. On the resulting set
we precompute Indyk's structure for distance 
$r' = O(\log^{1/p}n) \cdot r$. Given a query point $q$,
we embed the query point and query Indyk's structure. The space and runtime
follows.

For correctness, by the guarantees of Lemma \ref{lem:p-logn}, 
if $\|q-v\|_p \le r$ for some point $v \in V$, then under the expansion
guarantee of the mapping their $\ell_\infty$
distance is at most $2 (3 \ln n)^{1/p} \cdot r \le r'$, 
so the structure does not return {\em null}.
On the other hand, if the embedding succeeds then it is non-contractive,
and so any returned point must be within $\ell_p$ distance $O(\log\log d) \cdot r'$ of $q$. 
The probability follows from the contraction and expansion guarantees 
of Lemma \ref{lem:p-logn}.
\end{proof}

Finally, we use the near neighbor algorithm to solve the ANN problem, 
which was our ultimate goal:

\begin{theorem}\label{thm:p-logn}
Given set $V \subset \ell_p$ for $p>2$,
there exists a data structure of size $n^{1+\delta}$ (for arbitarily small constant $\delta$) 
which returns an 
$O(\log\log d \log^{1/p}n)$-ANN
in time $O(d \log^2 n) \cdot \big\lceil \frac{\log \log n}{p} \big\rceil$, 
and is correct with constant probability.
\end{theorem}

\begin{proof}
We invoke the reduction of Har-Peled \etal \cite{HIM12} to reduce ANN to 
$O(\log n)$ near neighbor queries.
We require that all $O(\log n)$ queries succeed with constant
probability, hence each near neighbor query must be correct with 
probability $1-O \left( \frac{1}{\log n} \right)$.
Each near neighbor query is resolved by preprocessing and querying 
$O(\lceil \log_{2^p} \log n \rceil) = 
O \left( \Big\lceil \frac{\log \log n}{p} \Big\rceil \right)$ 
independent structures of Corollary \ref{cor:p-logn}.
The probability that all these structures fail simultaneously is 
$2^{-p \cdot O((\log \log n)/p)} = O \left( \frac{1}{\log n} \right)$, 
and so at least one is correct with the desired probability. The runtime follows.

The reduction of \cite{HIM12} increases the space usage by a factor of 
$O(\log^2 n)$, and the additional oracles by a factor of 
$O \left( \Big\lceil \frac{\log \log n}{p} \Big\rceil \right)$,
but these increases are subsumed under the constant $\delta$ in the exponent.
\end{proof}

We note that when $p = \Omega(\log \log n)$, we recover the 
$O(\log \log d)$-approximation guarantees of Indyk's $\ell_\infty$ structure,
previously known to extend only to $p = \Omega(\log d)$).

\subsection{Embedding with distortion dependent on the doubling dimension}\label{sec:p-ddim}

Here we give an ANN algorithm whose approximation factor depends
on the doubling dimension, instead of the cardinality of the space.
We begin with a statement that applies only to nets. 
Our approach is motivated by a technique for low-dimensional 
Euclidean embeddings that appeared in \cite{IN07}.

\begin{lemma}\label{lem:net}
Let set $V \subset \ell_p$ have minimum inter-point distance 1, 
and let $q \in \ell_p$ be any query point. Let $\ddim \ge 2$ be the
doubling dimension of $V \cup \{q\}$, and fix any value $c \ge 4$.
For all $p \ge 1$, embedding $f_1$ (of Lemma \ref{lem:p-logn})
applied to set $V \cup \{q\}$ satisfies 
\begin{itemize}
\item
Contraction: Let $W \subset V$ include all points at distance at least
$h = c(8 \log c \cdot \ddim \ln \ddim)^{1/p}$ from $q$. Then the
distance from $q$ to $W$ does not contract to $c$ or less, 
$$\min_{v \in W} \| f_1(q) - f_1(v)\|_\infty > c,$$
with probability at least $1-\ddim^{-6\ddim}$.
\item
Expansion: For any pair $v,w \in V \cup \{q\}$, 
$$\|f(v)-f(w)\|_\infty \le 2 \|v-w\|_p$$
with probability at least $1-2^{-p}$.
\end{itemize}
\end{lemma}

\begin{proof}
Let $W_i \subset W$ include all points 
with distance to query point $q$ in the range $[2^i,2^{i+1})$.
Since $W$ has minimum inter-point distance 1 and diameter less that $2^{i+1}$,
Lemma \ref{lem:pack} implies that 
$|W_i| \le 2^{(i+2)\ddim}$.
Let $E_i$ be the bad event that $W_i$ contains any point $v \in W_i$ for which
$\|f_1(v)-f_1(q)\|_\infty \le c$.

Set $j = \log h$, so that all points in $W$ are found in sets
$W_{j+k}$ for integral $k \ge 0$.
For any point $v \in W_{j+k}$
the probability that the distance from $q$ to $v$ contracts to $c$ or less is
\begin{eqnarray*}
\Pr[\|f(v)-f(q)\|_\infty \le c]
&=&	e^{-(\|v-q\|_p/c)^p}					\\
&\le&	e^{-(2^{j+k}/c)^p}					\\
&=&	e^{-(2^k (8 \log c \cdot \ddim \ln \ddim)^{1/p})^p}	\\
&=&	\ddim^{-2^{kp+3} \log c \cdot \ddim}			\\	
&\le&	\ddim^{-2^{k+3} \log c \cdot \ddim}.			
\end{eqnarray*}
Hence, the probability of bad event $E_{j+k}$ is at most
\begin{eqnarray*}
\Pr[E_{j+k}]
&\le&	|W_{j+k}| \cdot \ddim^{-2^{k+3} \log c \cdot \ddim}							\\
&=&	2^{(j + k + 2)\ddim} \cdot \ddim^{-2^{k+3} \log c \cdot \ddim}						\\
&=&	2^{(k+2+\log c)\ddim} \cdot (8 \ddim \ln \ddim)^{\dim/p} \cdot \ddim^{-2^{k+3} \log c \cdot\ddim}	\\
&\le&	2^{(k+2+\log c)\ddim} \cdot (8 \ddim \ln \ddim)^{\dim} \cdot \ddim^{-2^{k+3} \log c \cdot\ddim}		\\
&<&	2^{(k+5+\log c)\ddim} \cdot \ddim^{(2-2^{k+3} \log c) \ddim}						\\
&<&	\ddim^{(k+7+(1-2^{k+3})\log c)\ddim}									\\
&\le&	\ddim^{(k+9-2^{k+4})\ddim}										\\
&\le&	\ddim^{-(6 + 2^k)\ddim}										
\end{eqnarray*}
The probability that any point in $W$ contracts to within distance
$c$ of $q$ is at most
$\sum_{k=0}^{\infty} E_{j+k} 
< \sum_{k=0}^{\infty} \ddim^{-(6 + 2^k)\ddim} 
< \ddim^{-6\ddim}$,
as claimed.

The expansion guarantee follows directly from Fact \ref{fact:max-stable}:
$\Pr[\|f(v)-f(q)\|_\infty \ge 2 \|v-q\|_p]
\le 2^{-p}.$
\end{proof}

As before, Lemma \ref{lem:net} can be used to solve the near neighbor problem:

\begin{corollary}\label{cor:net}
Given set $V \subset \ell_p$ for $p>2$ and a distance $r$,
there exists a data structure of size $n^{1+\delta}$ (for arbitarily small constant $\delta$)
which solves the
$O(\log\log d (\log \log \log d \ddim \log \ddim)^{1/p})$-approximate
near neighbor problem for distance $r$ in time
$O(d \log n)$,
and is correct with probability
$1-2^{-p}-\ddim^{-6\ddim}$.
\end{corollary}

\begin{proof}
Given a set $V$ and distance $r$, we preprocess the set by extracting an 
$r$-net, and then scaling down the magnitude of all vectors by $r$, so that 
the resulting set has minimum inter-point distance 1.
We then compute the embedding of Lemma \ref{lem:net} 
into $\ell_\infty$ for each net-vector, and precompute Indyk's $\ell_\infty$ 
structure for distance $4$. Given a query point $q$, we scale it down by $r$, 
embed it into $\ell_\infty$ using the same embedding as before, and query
Indyk's structure on distance 4. The space and runtime follows.

For correctness, let
$\|q-v\|_p \le r$ for some point $v \in V$.
After extracting the net, some net-point $w$ satisfied
$\|w-v\|_p \le r$, and so
$\|q-w\|_p \le 2r$.
After scaling, we have
$\|q-w\|_p \le 2$, 
and after applying the embedding into $\ell_\infty$,
$\|q-w\|_p \le 4$ with probability at least 
$1-2^{-p} > \frac{3}{4}$.
In this case Indyk's near neighbor structure must return a point
within distance $O(\log \log d)$ of $q$ in the embedded space 
(that is $\ell_\infty$). 
By the guarantees of Lemma \ref{lem:net} 
(taking $c = O(\log \log d)$), the distance from
the returned point to $q$ in the scaled $\ell_p$ space is at most
$O(\log \log d (\log \log \log d \ddim \log \ddim)^{1/p})$, 
and so it is an 
$O(\log \log d (\log \log \log d \ddim \log \ddim)^{1/p})$-approximate
near neighbor in the origin space.
\end{proof}

Similar to the derivation of Theorem \ref{thm:p-logn}, we have:

\begin{theorem}
Given set $V \subset \ell_p$ for $p>2$,
there exists a data structure of size $n^{1+t}$ (for arbitarily small constant $t$) 
which returns an 
$O(\log\log d (\log \log \log d \cdot \ddim \log \ddim)^{1/p})$-ANN
in time 
$O(d \log^2 n \log \log n)
\cdot \left( \Big \lceil \frac{\log \log n}{p} \Big \rceil + 
\Big \lceil \frac{\log \log n}{\ddim \log \ddim} \Big \rceil \right)$, 
and is correct with positive constant probability.
\end{theorem}

When $p = \Omega(\log \ddim + \log \log d)$, we recover the 
$O(\log \log d)$-approximation of Indyk, and this improves upon the 
$p = \Omega(\log \log n)$ guarantee of the previous section.

\section{ANN for $\ell_p$-space via embedding into $\ell_2$}\label{sec:mazur}
In this section, we show that an embedding from $\ell_p$ ($p>2$) into $\ell_2$
can be used to derived an $2^{O(p)}$-ANN in logarithmic query time.

We review the guarantees of the Mazur map below, and show it can be used 
as an embedding into $\ell_2$. We then solve the ANN problem in the embedded
space. This is however non-trivial, as the map incurs distortion that depends 
on the set diameter, a problem we address below.

\subsection{The Mazur map}
The Mazur map for the real valued vectors is defined as a mapping from $L^m_p$ to $L^m_q$, 
$1 \le q < p \le \infty$. The mapping of vector $v \in L_p$ is defined as
$$M(v) = \{ |v(0)|^{p/q}, |v(1)|^{p/q}, \ldots, |v(m-1)|^{p/q} \},$$ 
where $v(i)$ is the $i$-th coordinate of $v$.
The following theorem introduces a scaled Mazur map, and is adapted from \cite{BL00}.

\begin{theorem}
Let $x,y \in L_p$, $p < \infty$, be vectors such that $\|x\|_p,\|y\|_p \le C$.
The Mazur map for $1 \leq q < p$ scaled down by factor $\frac{p}{q}C^{p/q-1}$
fulfills the following:
\begin{itemize}
\item
Expansion: The mapping is non-expansive.
\item
Contraction: 
$\| M(x) - M(y) \|_q  \ge \frac{q}{p} (2C)^{1-p/q} \|x-y\|_p^{p/q}$.
\end{itemize}
\end{theorem}

The scaled Mazur map implies an embedding from $\ell_p$ ($p>2$) into $\ell_2$,
as in the following. (See also \cite[Lemma 7.6]{N14}, a significantly more general 
result.)

\begin{corollary}\label{cor:mazur}
Any subset $V \subset \ell_p$, $p<\infty$ with 
$\|x\|_p \le C$ for all $x \in V$ 
possesses an embedding 
$f:V \rightarrow l_2$ with the following properties for all $x,y \in V$:
\begin{itemize}
\item Expansion: The embedding $f$ is non-expansive.
\item Contraction: For $\|x-y\|_p = u$,
$$\|f(x)-f(y)\|_q \ge \frac{2}{p}(2C)^{1-p/2}  u^{p/2}.$$
\end{itemize}
\end{corollary}

\subsection{Nearest neighbor search via the Mazur map.}\label{sec:mazurnns}
Using the Mazur map, we can give an efficient algorithm for ANN for all $p>2$.
Recall that by definition, $\ddim = O(\log n)$. 
First we define the $c$-bounded near neighbor problem ($c$-BNN) for $c>1$ as follows:
For a $d$-dimensional set $V$ for which $\|x\|_p \le c$ for all $x \in S$, 
given a query point $q$:
\begin{itemize}
\item
If there is a point in $V$ within distance 1 of query $q$, return some point in
$V$ within distance $\frac{c}{9}$ of $q$.
\item
If there is no point in $V$ within distance 1 of query $q$, return {\em null} or
some point in $V$ within distance $\frac{c}{9}$ of $q$.
\end{itemize}
(The term $\frac{c}{9}$ was chosen to simplify the calculations below.)

\begin{lemma}\label{lem:bnn}
For $c=p18^{p/2}$, there exists a data structure for the 
$c$-bounded near neighbor problem for $V \subset \ell_p$, $p>2$,
that preprocesses $V$ in time and space $n^{O(1)}$, 
and resolves a query in time
$O(d \log n)$
with probability $1-n^{-O(1)}$.
\end{lemma}

\begin{proof}
The points are preprocessed by first applying the scaled Mazur map to embed $V$ 
into $\ell_2$ in time $O(dn)$. We then use the Johnson-Lindenstrauss (JL) transform \cite{JL84}
(or the fast JL transform \cite{AC06}) to reduce dimension to
$d' = O(\log n)$ 
with no expansion and contraction less than $\frac{1}{2}$, in time 
$O(dn \log n)$.
On the new space, we construct a data structure of size 
$2^{O(d')} = n^{O(1)}$ 
supporting Euclidean 2-approximate near neighbor queries in 
$O(d' \log n)$ 
time per query
\cite{HIM12,KOR98}.

Given a query point $q$, we apply the Mazur map and JL transform on the new point in time $O(d \log n)$, 
and use the resulting vector as a query for the $2$-approximate near neighbor algorithm on the
embedded space in time $O(d' \log n)$.
If the point $x$ returned by this search satisfies $\|q-x \|_p \le \frac{c}{9}$ then we return it, 
and otherwise we return {\em null}.

To show correctness: The Mazur map is non-expansive, as is the JL transform 
(which is correct with probability $1-n^{-O(1)}$). By Corollary \ref{cor:mazur},
the Mazur map ensures that inter-point distances of $\frac{c}{9}$ or greater map to at least 
$\frac{2}{p} (2c)^{1-p/2}(c/9)^{p/2} 
= \frac{2}{p} 2c 18^{-p/2} 
= \frac{2}{p} 2(p18^{p/2}) 18^{-p/2} 
= 4$, 
and then the contraction guarantee of the JL-transform implies that the distance in the embedded
Euclidean space is greater than 2. It follows that if $q$ possesses a neighbor in the original space
at distance 1 or less, the $2$-ANN in the embedded Euclidean space finds a neighbor at 
distance 2 in the embedded space and less than $\frac{c}{9}$ in the origin space. 
\end{proof}

We will show that an oracle solving $c$-BNN can be used as a subroutine for a data structure solving the 
$c$-approximate nearest neighbor problem. This parallels the classical reduction from ANN to 
the near neighbor problem utilized above. However, in order to minimize 
the distortion introduced by the Mazur map we must restrict the oracle to bounded diameter sets, 
and this results in a different reduction, adapted from \cite{KL04}.

\begin{theorem}\label{thm:nns-big} 
Let $V$ and $c$ be as in Lemma \ref{lem:bnn}. 
There exists a data structure for the $6c = 6p18^{p/2}$-ANN problem on $V$, which preprocesses 
$\min \{2^{O(\ddim)}n, O(n^2)\} 
\cdot 
\Big \lceil \frac{\log \log d}{p \ddim } \Big \rceil$
separate $c$-BNN oracles, 
each for a subset of $V$ of size at most 
$w = \min\{ c^{O(\ddim(S))}, O(n) \}$, in total time and space 
$O(w n d \log \log d)$. 
The structure resolves a query with constant probability or correctness by executing 
$$O \left(\log d \cdot \Big \lceil \frac{\log \log d}{p \ddim } \Big \rceil \right)$$ 
$c$-BNN oracle invocations, in total time 
$$O \left( 
d^2 \log n + d \log d \cdot p 
\ddim \cdot \Big \lceil \frac{\log \log d}{p \ddim } \Big \rceil 
\right).$$ 
\end{theorem}

\begin{proof}
We preprocess a hierarchy and net-tree for $V$.
Given query point $q$, we will seek a hierarchical point $w \in S_j$ 
which satisfies $\|w-q\|_p \le 3 \cdot 2^j$, for minimal $j$. 
Note that if such a point $w$ exists, then every hierarchical ancestor
$w' \in S_i$ $(i>j)$ of $w$ also satisfies $\|w'-q\|_p \le 3 \cdot 2^i$:
We have $\|w'-w\|_p \le 2 \cdot 2^i - 2 \cdot 2^j$, and so
$$
\| w' - q \|_p 
\le \| w'-w \|_p + \| w - q \|_p 
< (2 \cdot 2^i - 2 \cdot 2^j) + 3 \cdot 2^j
< 3 \cdot 2^i.
$$

To find $w$, we modify the navigating-net algorithm
of Krauthgamer and Lee \cite{KL04}: Beginning with 
the root point at the top level of the hierarchy, we descend down
the levels of the hierarchy, while maintaining at each level $S_i$ 
a single point of interest $t \in S_i$ satisfying $d(t,q) \le 3 \cdot 2^i$.  

Let $t \in S_i$ be the point of interest in level $S_i$.
Then the next point of interest $t' \in S_{i-1}$ satisfies
$\|t' - t\|_p \le \| t' - q \|_p + \|q-t\|_p \le 3 \cdot 2^i + 3\cdot 2^{i-1} = 4.5 \cdot 2^i$.
So to find $t'$ it suffices to search all points of $S_{i-1}$ 
within distance $4.5 \cdot 2^i$ of $t$. But there may be $2^{\Theta(\ddim)}$ such points,
and so we cannot afford a brute-force search on the set.
Instead, we utilize the $c$-BNN data structure of Lemma \ref{lem:bnn}:
For each net-point $t \in S_i$, preprocess a set $N(t,i)$ that includes all these candidate
points of $S_{i-1}$, as well as all their descendants in the hierarchical level 
$S_k$, $k = i- \lceil \log 4c \rceil$ within distance $4.5 \cdot 2^i$ of $t$. 
After translating $N(t,i)$ so that $t$ is the origin,
and scaling so that all points have magnitude at most $c$, 
we preprocess for $N(t,i)$ the $c$-BNN oracle of Lemma \ref{lem:bnn}.
(Note that $|N(t,i)| = c^{O(\ddim)} = 2^{O(p \ddim)}$.)
To find $t'$, execute a $c$-BNN query on $N(t,i)$ and $q$, and if the query 
returns a point $q'$, then set $t'$ to be the ancestor of $q'$ in $S_{i-1}$:
$$
\| t' - q \|_p 
\le \| t' - q' \|_p + \| q' - q \|_p 
< 2 \cdot 2^{i-1} + \frac{4.5 \cdot 2^i}{9}
= 3 \cdot 2^{i-1}.
$$

We terminate the algorithm in any of three events:
\begin{itemize}
\item
The root $t \in S_i$ does not satisfy $\|q-t\|_p \le 3 \cdot 2^i$.
In this case the root itself is at worst a 3-ANN of $q$, 
as the distance from all descendants of the root (that is,
all points) to $q$ is greater than 
$\|q-t\|_p - 2 \cdot 2^i$.
\item 
The algorithm locates a point $w \in S_0$ satisfying $\|w-q\| \le 3$.
Then the nearest neighbor of $q$ can be either $w$ or a point within
distance 6 of $w$. We execute a $c$-BNN query on the set of points
within distance $6$ of $w$. (That is, as above we translate $w$ to the 
origin, scale so that the maximum magnitude is $c$, and 
preprocess a query structure.)
If the query returns {\em null}, then there
is no point within distance $\frac{6}{c}$ of $q$, and so $w$ is a
$\frac{3}{6/c} = \frac{c}{2}$-ANN of $q$. If the query return some point, then that
point is within distance $\frac{6}{9} = \frac{2}{3}$ of $q$. Since the
minimum inter-point distance of the set is 1, there cannot be another point within 
distance $\frac{1}{3}$ of $q$, and so the returned point is a 2-ANN of $q$. 
\item 
A $c$-BNN query on some set $N(t,i)$ returns {\em null}.
As the diameter of $N(t,i)$ is at least $2^i$, this implies
that there is no point in $N(t,i)$ within distance 
$\frac{2^i}{c}$ of $q$. As all descendants of a point in 
$S_k$ ($k=i-\lceil \log 4c \rceil$) are within distance
$2 \cdot 2^k$ of their ancestor, the distance from $q$ to
all other points is at least 
$\frac{2^i}{c} - 2 \cdot 2^k 
= \frac{2^i}{c} - 2\frac{2^i}{4c} 
= \frac{2^i}{2c}$. 
It follows that $t$ is a 
$\frac{3 \cdot 2^i}{2^i/2c} = 6c$-ANN of $q$.
\end{itemize}

We conclude that the above algorithm returns the nearest neighbor stipulated by
the lemma. However, its runtime depends on the number of 
levels in the hierarchy, that is $O(\min \{ n,\Delta \})$. To remove this dependence, 
we first invoke the algorithm of Chan \cite{Ch98} to find in time 
$O(d^2 \log n)$ and high probability a $O(d^{3/2})$-ANN of $q$, called $q'$. 
We then locate the ancestor of $q'$ in level
$S_i$, $i = \lceil \log \|q-q'\|_p \rceil$, of the hierarchy in time $O(\log n)$,
which can be done easily using standard tree decomposition algorithms.
We assign this ancestor as our first point of interest $t$; note that we have
$$\|t-q\|_p 
\le \|t-q'\|_p + \|q'-q\|_p
< 2 \cdot 2^i + 2^i
= 3 \cdot 2^i,$$
so indeed $t$ is a valid point of interest. 
After descending $O(\log d)$ levels in the 
search, we reach radii that are smaller than the true distance from $q$ to its
nearest neighbor in $V$, and the search must terminate.

In order for the entire procedure to succeed with constant probability, 
we require the failure probability of each level $c$-BNN query to be
$O(1/\log d)$. Each $c$-BNN oracle consists of 
$O \left( \Big \lceil \frac{\log \log d}{p \ddim } \Big \rceil \right)$ 
structures of Lemma \ref{lem:bnn}, and (recalling that a query is executed
on a set of size $2^{O(p \ddim)}$) the probability that they all fail
simultaneously is 
$2^{-O(p \ddim \cdot (\log \log d)/ p \ddim)} = O(1/\log d)$.
The final space and runtime follow directly from the time and space required for 
building and querying the $c$-BNN oracles, each of size $\min \{ O(n), 2^{O(p \ddim)} \}$, 
plus the single level ancestor query.
\end{proof}

\paragraph{Comment.}
Note that the query time in linear in the doubling dimension, as opposed to the
exponential dependence common to ANN for doubling metrics.
We can extend this lemma by noting that once a $2^{O(p)}$-ANN is found, we can run
the standard navigating net algorithm to descend $O(p)$ additional levels and locate
a constant-factor ANN in time $p2^{O(\ddim)}$. We then search
$\veps^{-O(\ddim)}$ more points in a brute-force fashion and locate a $(1+\veps)$-ANN. 
So a $(1+\veps)$-ANN can be found with $p2^{O(\ddim)} + \veps^{-O(\ddim)}$
additional work.

\paragraph{Acknowledgements.}
We thank Sariel Har-Peled, Piotr Indyk, Robi Krauthgamer, 
Assaf Naor and Gideon Schechtman for helpful conversations.

\bibliography{bib-nns}

\end{document}